\documentclass[12pt,reqno]{amsart}

\usepackage{amsmath,amssymb,amsthm,mathtools}
\usepackage[margin=1in]{geometry}
\usepackage{needspace}
\usepackage{xcolor}
\usepackage{url}
\usepackage{hyperref}
\usepackage[protrusion=false]{microtype}

\allowdisplaybreaks[3]
\numberwithin{equation}{section}

\definecolor{linkblue}{RGB}{25,77,135}
\hypersetup{
colorlinks=true,
linkcolor=linkblue,
citecolor=linkblue,
urlcolor=linkblue,
unicode=true,
pdfencoding=auto,
pdftitle={Tight entropy contraction beyond detailed balance},
pdfauthor={Li Gao and Jingyu Guo},
pdfsubject={Quantum channels, conditional variance, and complete entropy contraction},
pdfkeywords={quantum channel,quantum Markov semigroup,complete modified logarithmic Sobolev inequality,GNS spectral gap,Pimsner-Popa index,relative entropy}
}

\newtheorem{theorem}{Theorem}[section]
\newtheorem{proposition}[theorem]{Proposition}
\newtheorem{lemma}[theorem]{Lemma}
\newtheorem{corollary}[theorem]{Corollary}

\theoremstyle{remark}
\newtheorem{remark}[theorem]{Remark}

\newcommand{\Tr}{\operatorname{Tr}}
\newcommand{\id}{\operatorname{id}}
\newcommand{\Ran}{\operatorname{Ran}}
\newcommand{\Ker}{\operatorname{Ker}}
\newcommand{\spec}{\operatorname{spec}}
\newcommand{\EP}{\operatorname{EP}}
\newcommand{\cb}{\mathrm{cb}}

\newcommand{\cB}{\mathcal B}
\newcommand{\cD}{\mathcal D}
\newcommand{\cH}{\mathcal H}
\newcommand{\cK}{\mathcal K}
\newcommand{\cL}{\mathcal L}
\newcommand{\cN}{\mathcal N}
\newcommand{\cR}{\mathcal R}

\newcommand{\eps}{\varepsilon}
\newcommand{\dd}{\,\mathrm d}
\newcommand{\ip}[2]{\left\langle #1,#2\right\rangle}
\newcommand{\norm}[1]{\left\lVert #1\right\rVert}

\newcommand{\GNS}{\mathrm{GNS}}

\newcommand{\Id}{\mathrm I}
\newcommand{\Var}{\operatorname{Var}}
\newcommand{\HS}{\mathrm{HS}}

\title[Entropy contraction beyond detailed balance]{Tight entropy contraction beyond detailed balance}
\author{Li Gao}
\address{School of Mathematics and Statistics, Wuhan University, Wuhan 430072, China}
\address{Wuhan Institute of Quantum Technology, Wuhan 430075, China}
\email{gao.li@whu.edu.cn}
\author{Jingyu Guo}
\address{School of Mathematics and Statistics, Wuhan University, Wuhan 430072, China}
\email{guojingyu@whu.edu.cn}
\date{September 21, 2026}

\begin{document}
\begin{abstract}
We establish quantitative entropy contraction bounds for finite-dimensional
quantum channels beyond detailed balance. For channels compatible with a
faithful conditional expectation, we show that GNS (Gelfand--Naimark--Segal) norm contraction yields
relative entropy
contraction with only a logarithmic loss in the dimensional constant. For quantum Markov semigroups
with a faithful asymptotic conditional
expectation, this gives a lower bound on the complete modified logarithmic
Sobolev constant in terms of the GNS spectral gap.
The bounds allow initial entanglement with an external environment and are independent
of its system size. Applications to
alternating reservoir pulses and continuously coupled energy ladders give
bounds with optimal size dependence.
\end{abstract}

\maketitle

\section{Introduction}

Relative entropy is a fundamental measure of the distinguishability
of quantum states.
Its decrease under a quantum channel quantifies information loss
in quantum information processing.
This monotonicity is called the data-processing inequality
\cite{Lindblad1975}, and strong data-processing inequalities
quantify the contraction \cite{HircheRouzeFranca2022}, which has wide applications in quantum
information and quantum computation.
Entropy contraction estimates yield mixing-time bounds for quantum
Markov semigroups \cite{KastoryanoTemme2013}, establish rapid
thermalization of commuting quantum spin chains
\cite{BardetCapelGaoLuciaPerezGarciaRouze2023}, and constrain the
performance of noisy quantum optimization algorithms
\cite{FrancaGarciaPatron2021}.
Mathematically, quadratic norm estimates are often easier to obtain,
but converting them into entropy contraction can introduce a dependence
on the dimension of the system.
In this work, we prove a conversion with only logarithmic dependence
on the index of the equilibrium projection, without detailed balance.
The estimate also extends to the presence of an idle external system,
uniformly over its dimension and allowing initial entanglement.

Let $\Phi:\cB(\cH)\to\cB(\cH)$ be a unital and completely positive map.
We use the Heisenberg picture: $\Phi$ acts on observables and its trace
adjoint $\Phi_*$ is a quantum channel (completely positive and trace preserving) acting on states.
Throughout the paper, we restrict the discussion to finite dimensions. Let $E$ be a faithful
conditional expectation
$E:\cB(\cH)\to\cN$ onto the fixed point subalgebra $\cN$ satisfying
\begin{align}
\Phi E=E\Phi=E.
\label{eq:channel-compatibility}
\end{align}
Thus $\Phi$ fixes $\cN$ pointwise and preserves the reference state
$E_*\rho$ associated with each initial state $\rho$. We allow nontrivial
fixed algebras $\cN$ which means $\Phi$ admits more than one equilibrium state. Recall that for two
quantum states $\rho$ and $\sigma$, the relative entropy is
\begin{align}
D(\rho\Vert\sigma)&=\Tr\rho(\log\rho-\log\sigma).
\end{align}
The ordinary and complete entropy contraction coefficients are
\begin{align}
\eta(\Phi\mid E)
 &=\sup_{\rho:\,D(\rho\Vert E_*\rho)>0}
 \frac{D(\Phi_*\rho\Vert E_*\rho)}{D(\rho\Vert E_*\rho)},
 \qquad
 \eta^{\mathrm c}(\Phi\mid E)
 =\sup_{n\ge1}\eta(\Phi\otimes \id_{M_n}\mid E\otimes \id_{M_n}).
\end{align}
The first supremum is over states, while the complete version includes all joint
states, with reference $(E_*\otimes \id_{M_n})\rho$ retaining their external marginal.
For these coefficients in the GNS-symmetric setting, see
\cite[Definition~5.1]{GaoJungeLaRacuenteLi2025}.

Fix a faithful state $\omega$ with $E_*\omega=\omega$. Define the GNS (Gelfand--Naimark--Segal)
inner product and corresponding contraction coefficient
\begin{align}
\ip{X}{Y}_{2,\omega}&=\Tr(\omega X^*Y),\qquad
s=\norm{\Phi-E}_{2,\omega\to2,\omega}
\end{align}
Proposition~\ref{prop:channel-stability} shows that $s\in[0,1]$ is independent
of $\omega$ and unchanged by matrix amplification.
Recall the Pimsner--Popa index \cite[Section~2]{PimsnerPopa1986} and its complete
version as
\begin{align}
C(E)&=\inf\{C\ge1:\rho\le CE_*\rho\text{ for all }\rho\ge0\},
\qquad C_{\cb}(E)=\sup_{n\ge1}C(E\otimes \id_{M_n}).
\label{eq:indices}
\end{align}

\begin{theorem}[Entropy contraction]
\label{thm:channel}
Let $E$ be a faithful conditional expectation and
$\Phi$ be a unital completely positive map with $\Phi E=E\Phi=E$. Let $\omega=E_*\omega$ be a faithful invariant state and suppose the GNS contraction coefficient
\[s=\norm{\Phi-E}_{2,\omega\to2,\omega} <1 \]
Then for any state $\rho$,
\begin{align}
D(\Phi_*\rho\Vert E_*\rho )\le \eta D(\rho\Vert E_*\rho )
\label{eq:channel-main}
\end{align}
with the optimal contraction coefficient
\begin{align}
\eta &\le\left(1+\frac{c_0(1-s^2)}{1+\log C(E)}\right)^{-1},
 \qquad c_0=\frac{1-\log2}{4}.
\end{align}
The same bound for the complete coefficient $\eta^c$ holds with $C(E)$
replaced by $C_{\cb}(E)$.
\end{theorem}
We emphasize that no detailed balance assumption is required above.
Under GNS symmetry, complete entropy contraction with logarithmic
index dependence was obtained in
\cite[Theorem~4.8 and Corollary~5.2]{GaoJungeLaRacuenteLi2025}.
Our Theorem~\ref{thm:channel} retains this dependence for general
faithful invariant states without detailed balance.

The entropy contraction also applies to a quantum Markov semigroup
(QMS) in continuous time setting. Let $P_t=e^{t\cL}$ be QMS admit a faithful asymptotic conditional expectation
$\displaystyle E=\lim_{t\to \infty} P_t$.
The modified logarithmic Sobolev inequality (MLSI) characterizes the exponential decay of relative entropy
\begin{align}
D(P_{t*}\rho\Vert E_*\rho)
 &\le e^{-2\alpha t}D(\rho\Vert E_*\rho),\qquad t\ge0,
\label{eq:intro-decay}
\end{align}
The optimal (largest) MLSI constant is denoted as $\alpha(\cL)$.
The complete modified logarithmic Sobolev inequality (CMLSI) requires the
same estimate after every matrix amplification, with optimal constant
$\alpha^{\mathrm c}(\cL)=\inf_{n\ge1}\alpha(\cL^{(n)})$;
see \cite[Definition~2.8(c)]{GaoJungeLaRacuente2020}.
For a faithful invariant state $\omega=E_*\omega$, the GNS spectral gap is
\begin{align}
\lambda_{\GNS}(\cL)
 &=\inf_{\substack{X\ne0\\E(X)=0}}
 \frac{-\Re\Tr(\omega X^*\cL(X))}{\Tr(\omega X^*X)}.
\end{align}
This is indeed the gap
of the GNS symmetrization relative to $\Ran E$, as noted in
\cite[Remark~3.9]{Wirth2026}, and more importantly, $\lambda_{\GNS}$ is independent of $\omega$ and stable
under matrix amplification.

The following is a continuous time version of Theorem~\ref{thm:channel}.
\begin{corollary}[MLSI bound]
\label{cor:main}
Let $P_t=e^{\cL t}$ be a quantum Markov semigroup with a faithful asymptotic conditional expectation
$\displaystyle E=\lim_{t\to \infty} P_t$. Suppose it has a positive GNS gap $\lambda_{\GNS}(\cL)>0$.
Then for any state $\rho$ and $t\ge 0$,
\[ D(P_{t*}\rho \Vert E_*\rho )\le e^{-2\alpha t} D(\rho\Vert E_*\rho  )\ ,\]
where the optimal constant $\alpha$ satisfies
\begin{align}
\alpha&\ge
\frac{c_0\lambda_{\GNS}(\cL)}{1+\log C(E)},
\end{align}
and $C(E)$ is the Pimsner-Popa index \eqref{eq:indices}.
The same bound holds for the amplification semigroup $P_t\otimes \id_{M_n}$ with $C(E)$
replaced by $C_{\cb}(E)$, uniformly for $n\ge 1$.
\end{corollary}

Under GNS symmetry, MLSI bounds with inverse index were obtained by Gao and Rouz\'e in
\cite{GaoRouze2022}, and logarithmic bounds with optimal index
dependence by Gao, Junge, LaRacuente and Li in \cite{GaoJungeLaRacuenteLi2025}.
Without detailed balance, the authors
proved a complete MLSI lower bound proportional to
$\lambda_{\GNS}(\cL)/C_{\cb}(E)$ \cite{GaoGuo2026}, analogous to \cite{GaoRouze2022} in the GNS-symmetric setting.
Our Corollary~\ref{cor:main} replaces this linear index loss by a logarithmic
one, which establishes the optimal index
dependence beyond the symmetric condition. For primitive KMS-symmetric semigroups, Liu, Wan and Wu
provide an MLSI bound quadratic in BKM coercivity \cite{LiuWanWu2026}.
For general primitive semigroups, a hypercontractive approach to MLSI bounds with logarithmic
index dependence was obtained in \cite{GaoWang2026}. Finite-time complete entropy decay for
bistochastic QMS was studied in \cite{LaRacuente2026}.

Corollary~\ref{cor:main} follows from iterating the one time estimate \eqref{eq:channel-main} and
taking the infinitesimal limit.
The core and mysterious ingredient in the proof of Theorem~\ref{thm:channel} is Lemma~\ref{lem:centering}.
For $\rho\ge aE_*\rho$, it shows that centering the
resolvent variance by $E$ retains at least an $a$ fraction of $D(E_*\rho\Vert\rho)$. This conditional variance
bound is combined with the nonnegative remainders in the channel monotonicity argument
of Lesniewski and Ruskai \cite[Theorems~2.5 and~2.15]{LesniewskiRuskai1999};
see also \cite[Sections~2.3 and~4.3]{Ruskai2007}.
The resulting entropy loss controls the reverse relative entropy $D(E_*\rho\Vert \Phi_*\rho)$ of the
output. The
index and the entropy comparisons of Gao and Zhao
\cite[Proposition~2.3]{GaoZhao2026} convert this estimate into,
entropy contraction, where the logarithmic loss enters only at this last comparison.

As applications, we discuss quantum channels and Markov semigroups from local transitions on an
energy ladder coupled to
squeezed reservoirs. Alternating the even and odd transitions gives a
channel whose ordinary and complete entropy contraction deficits are both
of order $d^{-1}$ for $d$-level system.
For continuous coupling to the same reservoirs, the ordinary and complete
MLSI constants are of order $d^{-1}$, uniformly up to ideal
squeezing (Proposition~\ref{prop:squeezed-ladder}).
Both applications include upper bounds with matching order.

Section~\ref{sec:preliminaries} collects the preliminaries on conditional expectations,
GNS contraction, and quantum relative entropy. Section~\ref{sec:channels} proves the key conditional
variance estimate Lemma~\ref{lem:centering} and Theorem~\ref{thm:channel}.
Section~\ref{sec:semigroups} derives the consequences on quantum Markov semigroup, and
Section~\ref{sec:applications} discusses examples of two reservoir models.

\noindent\textbf{AI Statement.}  The authors acknowledge assistance from AI tools during the
exploratory discussion and manuscript writing of this project. The main idea of the proof argument,
including the key Lemma~\ref{lem:centering}, were suggested by ChatGPT (GPT-6 Astra).
All mathematical arguments and proofs in the final manuscript were verified and written by the authors.

\noindent\textbf{Acknowledgement.}  Li Gao and Jingyu Guo are
supported in part by the National Natural Science Foundation of China (grant no.
12401163) and the Department of Science and Technology of Hubei Province (project nos.
2025EHA041 and 2025AFA044).
\section{Preliminaries}
\label{sec:preliminaries}

\subsection{Conditional expectations and relative entropy}
Throughout the paper, all Hilbert spaces are finite-dimensional, and $\Tr$ denotes the standard
matrix trace. For a linear map $T$ on $\cB(\cH)$, its trace adjoint is
determined by
\begin{align*}
\Tr[T_*(\rho)X]&=\Tr[\rho T(X)].
\end{align*}
Thus, a unital completely positive map $\Phi$ acts on observables, and its
trace pre-adjoint $\Phi_*$ is a trace-preserving quantum channel on states.
A conditional expectation onto a unital $*$-subalgebra
$\cN\subseteq\cB(\cH)$ is a unital completely positive projection
$E:\cB(\cH)\to\cN$ satisfying $E\circ E=E$. $E$ is called faithful if $E(X^*X)=0$ implies $X=0$,
and satisfies
\begin{align*}
E(AXB)&=AE(X)B,\qquad A,B\in\cN.
\end{align*}
For a linear map $T$, write
\begin{align}
T^{(n)}&=T\otimes\id_{M_n},\qquad M_n=\cB(\mathbb C^n).
\end{align}
Recall the Pimsner--Popa indices from \eqref{eq:indices}:
\begin{align*}
C(E)&=\inf\{C\ge1:\rho\le CE_*\rho\ \text{ for all }\ \rho\ge0\},
\qquad C_{\cb}(E)=\sup_{n\ge1}C(E\otimes \id_{M_n}).
\end{align*}
The orthogonal decomposition of a faithful conditional expectation and the
associated index formula
take the following form \cite[Proposition~3.2]{GaoZhao2026}.

\begin{proposition}[Structure of conditional expectations and index]

For a faithful conditional expectation $E:\cB(\cH)\to\cN$, there are
finite-dimensional spaces $\cK_k,\cR_k$ and faithful density matrices
$\sigma_k\in\cB(\cR_k)$ such that, up to unitary equivalence,
\begin{align}
\cH&=\bigoplus_k\cK_k\otimes\cR_k,
&\cN&=\bigoplus_k\cB(\cK_k)\otimes \mathbb{C}\Id_{\cR_k}.
\label{eq:blocks}
\end{align}
Writing $p_k$ for the block projections onto $\cK_k\otimes\cR_k$,
$\varphi_{\sigma_k}(Z)=\Tr(\sigma_kZ)$, and
$\eta_k=\Tr_{\cR_k}(p_k\rho p_k)$, we have
\begin{align}
E(X)&=\bigoplus_k
 \bigl[(\id_{\cB(\cK_k)}\otimes\varphi_{\sigma_k})(p_kXp_k)\bigr]
 \otimes\Id_{\cR_k},\\
E_*\rho&=\bigoplus_k\eta_k\otimes\sigma_k.
\end{align}
Consequently, the $E_*$-invariant states are exactly the states of the
latter form; such a state is faithful if and only if every $\eta_k>0$.
Moreover,
\begin{align}
C(E)\leq C_{\cb}(E)&=\sum_k\Tr(\sigma_k^{-1})<\infty.
\end{align}
\end{proposition}
The order bound from finite index ensures that $D(\rho\Vert E_*\rho)$ is finite for
every state. For a faithful $E_*$-invariant state $\omega$, the chain
rule is
\begin{align}
D(\rho\Vert\omega)
 &=D(\rho\Vert E_*\rho)+D(E_*\rho\Vert\omega).
\end{align}
In particular, $\rho\mapsto D(\rho\Vert E_*\rho)$ is continuous on the entire state space. Thus
estimates proved for
faithful states passes to arbitrary states through perturbation
$\rho_\eps=(1-\eps)\rho+\eps\omega$.

\subsection{GNS contraction and conditional variance}

Let $E_*\omega=\omega>0$ be a faithful invariant state. For the GNS inner product
$\ip{X}{Y}_{2,\omega}=\Tr(\omega X^*Y)$, we define the conditional variance with respect to $\omega$ condition on the subalgebra $\mathcal{N}$
\begin{align}
\Var_{E,\omega}(X)&=\norm{X-EX}_{2,\omega}^2.
\end{align}
The bimodular property and invariance give, for $A\in\cN$,
\begin{align*}
\ip{A}{EX}_{2,\omega}
 &=\Tr[\omega E(A^*X)]
 =\Tr(\omega A^*X)
 =\ip{A}{X}_{2,\omega}.
\end{align*}
Hence $E$ is the GNS orthogonal projection onto $\cN$, and
\begin{align}
\Tr(\omega X^*X)
 &=\Var_{E,\omega}(X)+\Tr\big(\omega(EX)^*EX\big),\nonumber\\
\Tr(\omega XX^*)
 &=\Var_{E,\omega}(X^*)+\Tr\big(\omega(EX)(EX)^*\big).
\label{eq:orthogonal-decompositions}
\end{align}

\begin{proposition}[State and amplification invariance]
\label{prop:channel-stability}
Let $\Phi$ be unital and completely positive, with
$\Phi E=E\Phi=E$. Then
\begin{align*}
s&=\norm{\Phi-E}_{2,\omega\to2,\omega}\le 1
\end{align*}
is independent of the faithful $E_*$-invariant state $\omega$. Equivalently, for $\kappa=1-s^2$,
\begin{align}
\Tr\bigl[\omega\bigl(\Phi(X^*X)-\Phi(X)^*\Phi(X)\bigr)\bigr]
 &\ge\kappa\Var_{E,\omega}(X).
\label{eq:channel-variance}
\end{align}
Moreover, both $s$ and $\kappa$ are unchanged under matrix amplification,
\[ s =\norm{\Phi\otimes \id_{M_n}-E\otimes\id_{M_n}}_{2,\omega_n \to2,\omega_n}\]
for arbitrary faithful
states satisfying $(E_*\otimes\id_{M_n})(\omega_n)=\omega_n$.
\end{proposition}

\begin{proof}
The case $E=\id$ is immediate, so assume $E\ne\id$.
Since $\Phi_*\omega=\Phi_*E_*\omega=E_*\omega=\omega$, Kadison-Schwarz inequality gives
\begin{align*}
\norm{\Phi X}_{2,\omega}^2=\Tr[\omega\Phi(X)^*\Phi(X)]
 &\le\Tr[\omega\Phi(X^*X)]
 =\Tr[(\Phi_*\omega)X^*X]
 =\norm{X}_{2,\omega}^2.
\end{align*}
Thus $s\le1$. Since $EX=E\Phi(X)$, by orthogonality
\begin{align*}
&\Tr\bigl[\omega\bigl(\Phi(X^*X)-\Phi(X)^*\Phi(X)\bigr)\bigr]\\
 =&\norm{X}_{2,\omega}^2-\norm{\Phi X}_{2,\omega}^2\\
 =&\norm{X-EX}_{2,\omega}^2+\norm{EX}_{2,\omega}^2-\norm{\Phi X-EX}_{2,\omega}^2-\norm{EX}_{2,\omega}^2\\
 =&\norm{X-EX}_{2,\omega}^2-\norm{\Phi X-EX}_{2,\omega}^2\\
 \ge&(1-s^2)\norm{X-EX}_{2,\omega}^2
 =\kappa\Var_{E,\omega}(X).
\end{align*}
For the state invariance, we follow the similar argument for Lindbladians \cite[Sections~4 and~5]{GaoGuo2026}.
Since $\cN$ lies in the multiplicative domain of $\Phi$, so $\Phi$ is
$\cN$-bimodular. Following \eqref{eq:blocks}, write
\begin{align*}
\Phi\big|_{\cB(\cK_l,\cK_k)\otimes\cB(\cR_l,\cR_k)}
 =&\id\otimes\Phi_{kl},\\
\mathcal V_{kl}
 =&\begin{cases}
 \cB(\cR_l,\cR_k),&k\ne l,\\
 \{Z\in\cB(\cR_k):\Tr(\sigma_kZ)=0\},&k=l,
 \end{cases}\\
 \norm{Z}_{kl}^2=&\Tr(\sigma_lZ^*Z)\ ,\  Z\in\cB(\cR_l,\cR_k).
\end{align*}
Here $\norm{\cdot}_{kl}$ is the GNS norm on $\cB(\cR_l,\cR_k)$ with respect to the density $\sigma_l$.
Diagonalize each $\eta_l=\sum_{j} \lambda_{lj} |\phi_{lj} \rangle \langle\phi_{lj}|$ in
$\omega=\bigoplus_l\eta_l\otimes\sigma_l$, with eigenvalues
$\lambda_{lj}>0$.
Write
\[
X=\sum_{k,l,i,j}|\phi_{ki}\rangle\langle\phi_{lj}|
  \otimes X_{ki,lj},
\qquad X_{ki,lj}\in\cB(\cR_l,\cR_k).
\]
These components are GNS-orthogonal, with
\[
\norm{|\phi_{ki}\rangle\langle\phi_{lj}|\otimes Z}_{2,\omega}^2
=\lambda_{lj}\Tr(\sigma_lZ^*Z),
\qquad Z\in\cB(\cR_l,\cR_k).
\]
Moreover, $E$ annihilates the off-diagonal corners and applies
$Z\mapsto\Tr(\sigma_kZ)\Id_{\cR_k}$ to each entry of the
$(k,k)$ corner. Hence
\begin{align}
\norm{X}_{2,\omega}^2
 =&\sum_{k,l,i,j}\lambda_{lj}\norm{X_{ki,lj}}_{kl}^2,
\\
EX=&0\quad\Longleftrightarrow\quad
X_{ki,lj}\in\mathcal V_{kl}\text{ for all }k,l,i,j.\nonumber
\end{align}
It follows that
\begin{align*}
s^2
 &=\sup_{\substack{X\ne0\\EX=0}}
   \frac{\sum_{k,l,i,j}\lambda_{lj}
       \norm{\Phi_{kl}(X_{ki,lj})}_{kl}^2}
        {\sum_{k,l,i,j}\lambda_{lj}\norm{X_{ki,lj}}_{kl}^2}\\
 &=\max_{k,l:\,\mathcal V_{kl}\ne\{0\}}
   \norm{\Phi_{kl}|_{\mathcal V_{kl}}}_{kl\to kl}^2.
\end{align*}
This expression does not involve the $\eta_l$, which implies the states invariance. It is also unchanged under
amplification, which only replaces $\cK_k$ by $\mathbb C^n\otimes\cK_k$
and leaves all component norms $\|\cdot\|_{kl}$ unchanged.
\end{proof}

\subsection{Resolvent representations of relative entropy}

We recall the relative $f$-entropy formula of Lesniewski and
Ruskai from \cite{LesniewskiRuskai1999}. Fix faithful states $\rho,\sigma$.
On Hilbert--Schmidt space, with
$\ip{X}{Y}_{\HS}=\Tr(X^*Y)$, write
\begin{align*}
\mathsf L_B(X)&=BX,
&\mathsf R_B(X)&=XB,
&\mathsf{M}_u&=\mathsf L_\rho+u\mathsf R_\sigma.
\end{align*}
For $u>0$, $\mathsf M_u$ is positive definite. Define
\begin{align}
Z_u&=\mathsf{M}_u^{-1}(\rho-\sigma),
&\mathcal Q_u&=\ip{\rho-\sigma}{Z_u}_{\HS}=\ip{\rho-\sigma}{\mathsf{M}_u^{-1}(\rho-\sigma)}_{\HS}.
\label{eq:resolvent-definitions}
\end{align}
Equivalently, $Z_u$ is the unique solution of
\begin{align}
\rho Z_u+uZ_u\sigma&=\rho-\sigma.
\end{align}
The variational formula is
\begin{align}
\mathcal Q_u=&\ip{Z_u}{\mathsf{M}_u Z_u}_{\HS}
 =\Tr(\rho Z_uZ_u^*)+u\Tr(\sigma Z_u^*Z_u)\ge0,
\\
\mathcal Q_u=&        \max_Z\bigl\{2\Re\ip{\mathsf{M}_u Z_u}{Z}_{\HS}
       -\ip{Z}{\mathsf{M}_u Z}_{\HS}\bigr\}\nonumber\\
 =&\max_Z\bigl\{2\Re\ip{\rho-\sigma}{Z}_{\HS}
       -\Tr(\rho ZZ^*)-u\Tr(\sigma Z^*Z)\bigr\}.
\label{eq:variational}
\end{align}
Indeed, the completion of squares used in
\cite[Section~2.4]{Ruskai2007} reads
\begin{align*}
2\Re\ip{\rho-\sigma}{Z}_{\HS}-\ip{Z}{\mathsf{M}_uZ}_{\HS}
 &=\mathcal Q_u-\ip{Z-Z_u}{\mathsf{M}_u(Z-Z_u)}_{\HS}.
\end{align*}
The maximizer is $Z_u$, which need not be self-adjoint.

The relative entropy representations below are logarithmic
specializations from \cite[Definition~2.1 and Theorem~2.5]{LesniewskiRuskai1999}.
For an operator convex function $f$ with $f(1)=0$, their notation is
\begin{align*}
H_f(\sigma,\rho)
 =&\ip{\sigma^{1/2}}{f(\Delta_{\rho,\sigma})(\sigma^{1/2})}_{\HS}, &\Delta_{\rho,\sigma}&=\mathsf L_\rho\mathsf R_\sigma^{-1}, \\
H_{t\log t}(\sigma,\rho)&=D(\rho\Vert\sigma),&\\
H_{-\log t}(\sigma,\rho)&=D(\sigma\Vert\rho),&\\
H_{h_u}(\sigma,\rho)&=\mathcal Q_u, &h_u(t)&=\frac{(t-1)^2}{t+u},
\end{align*}
The reverse formula $D(\sigma\Vert\rho)$ is also \cite[Eq.~(7)]{Ruskai2007} with
$P=\sigma$, $Q=\rho$. We derive its equivalent energy form
\eqref{eq:reverse-energy}, whose two quadratic terms are weighted by
$\sigma$ and can therefore be centered using
\eqref{eq:orthogonal-decompositions}.

\begin{lemma}[Resolvent representations]

For faithful states $\rho$ and $\sigma$,
\begin{align}
D(\rho\Vert\sigma)
 &=\int_0^\infty\frac{u}{(1+u)^2}\mathcal Q_u\dd u,
\label{eq:forward-resolvent}\\
D(\sigma\Vert\rho)
 &=\int_0^\infty\frac{1}{(1+u)^2}\mathcal Q_u\dd u
\label{eq:reverse-resolvent}\\
 &=\int_0^\infty\frac{u}{(1+u)^2}
   \bigl[\Tr(\sigma Z_uZ_u^*)+u\Tr(\sigma Z_u^*Z_u)\bigr]\dd u.
\label{eq:reverse-energy}
\end{align}
\end{lemma}

\begin{proof} Note that
$H_{t-1}(\sigma,\rho)=\Tr(\rho-\sigma)=0$, and
\[ H_{t\log t}(\sigma,\rho)=H_{t\log t-t+1}(\sigma,\rho)\ , \ H_{-\log t}(\sigma,\rho)=H_{t-1-\log t}(\sigma,\rho)\ .\]
Then \eqref{eq:forward-resolvent}--\eqref{eq:reverse-resolvent} follows applying spectral calculus
to the integral representations
\begin{align*}
t\log t-t+1
 &=\int_0^\infty\frac{u\,h_u(t)}{(1+u)^2}\dd u,\\
t-1-\log t
 &=\int_0^\infty\frac{h_u(t)}{(1+u)^2}\dd u.
\end{align*}
For \eqref{eq:reverse-energy}, set
$\rho_t=t\rho+(1-t)\sigma$ and differentiate at $t=1$.
With $A_t=\rho_t-\sigma, A=\rho-\sigma$ and
$\mathsf{M}_u(t)=\mathsf L_{\rho_t}+u\mathsf R_\sigma$,
\begin{align*}
\dot A&=A,\\
\dot{\mathsf{M}}_u&=\mathsf L_A,\\
\frac{\dd}{\dd t}\mathsf{M}_u(t)^{-1}\Big|_{t=1}
 &=-\mathsf{M}_u^{-1}\dot{\mathsf{M}}_u \mathsf{M}_u^{-1},\\
\dot{\mathcal Q}_u
 &=2\Re\ip{\dot A}{Z_u}_{\HS}
   -\ip{Z_u}{\dot{\mathsf{M}}_u(Z_u)}_{\HS}\\
 &=2\mathcal Q_u-\Tr(AZ_uZ_u^*).
\end{align*}
Therefore
\begin{align}
\dot{\mathcal Q}_u
 &=\mathcal Q_u+\Tr(\sigma Z_uZ_u^*)
                   +u\Tr(\sigma Z_u^*Z_u).
\label{eq:mixing-derivative}
\end{align}
On the other hand,
\begin{align*}
\frac{\dd}{\dd t}D(\rho_t\Vert\sigma)\Big|_{t=1}
 &=\Tr[(\rho-\sigma)(\log\rho-\log\sigma)]\\
 &=D(\rho\Vert\sigma)+D(\sigma\Vert\rho).
\end{align*}
Faithfulness gives $\rho_t,\sigma\ge m\Id$ for some $m>0$, hence
\begin{align*}
\norm{\mathsf{M}_u(t)^{-1}}_{\HS\to\HS}
 &\le\frac1{m(1+u)},
&\frac{u}{(1+u)^2}|\partial_t\mathcal Q_u(t)|
 &\le\frac{K u}{(1+u)^3}.
\end{align*}
The latter bound is integrable. Differentiating
\eqref{eq:forward-resolvent} under the integral and using
\eqref{eq:mixing-derivative} gives
\begin{align*}
D(\sigma\Vert\rho)&=\frac{\dd}{\dd t}D(\rho_t\Vert\sigma)\Big|_{t=1}-D(\rho\Vert\sigma)\\
 &=\int_0^\infty\frac{u}{(1+u)^2}
       [\dot{\mathcal Q}_u-\mathcal Q_u]\dd u\\
 &=\int_0^\infty\frac{u}{(1+u)^2}
       [\Tr(\sigma Z_uZ_u^*)+u\Tr(\sigma Z_u^*Z_u)]\dd u. \qedhere
\end{align*}
\end{proof}

\subsection{Entropy comparisons and convexity}

We use the following consequences from
\cite[Proposition~2.3]{GaoZhao2026}.
\begin{lemma}[Entropy comparisons]
\label{lem:entropy-comparisons}
For faithful states $\rho,\sigma$, if $\rho\le C\sigma$ with $C>1$,
then
\begin{align}
D(\rho\Vert\sigma)
 &\le(1+\log C)D(\sigma\Vert\rho).
\end{align}
For $0<\theta<1$ and $\rho_\theta=(1-\theta)\sigma+\theta\rho$,
\begin{align}
D(\rho_\theta\Vert\sigma)
 &\ge\left(\theta+(1-\theta)\log(1-\theta)\right)D(\rho\Vert\sigma).
\end{align}
\end{lemma}

\begin{proof}
The \cite[Proposition~2.3]{GaoZhao2026} gives
\begin{align*}
D(\rho\Vert\sigma) &
 \le\frac{C\log C-C+1}{C-1-\log C}D(\sigma\Vert\rho)
 \le(1+\log C)D(\sigma\Vert\rho),\\
D((1-\theta)\sigma+\theta\rho\Vert\sigma) &
 \ge\left(\theta+(1-\theta)\log(1-\theta)\right)D(\rho\Vert\sigma)
\end{align*}
where in the first estimate we used the fact
$2(e^x-1-x)\ge x^2$ for $x=\log C>0$.
\end{proof}
\begin{lemma}
\label{lem:ep-convexity}
Let $T$ be unital and completely positive, and let
$T_*\sigma=\sigma>0$. Then
\begin{align*}
\Delta_{T,\sigma}(\rho)
 &=D(\rho\Vert\sigma)-D(T_*\rho\Vert\sigma)
\end{align*}
is convex on the entire state space. In particular,
\begin{align*}
\Delta_{T,\sigma}((1-\theta)\sigma+\theta\rho)
 &\le\theta\Delta_{T,\sigma}(\rho),\qquad 0\le\theta\le1.
\end{align*}
\end{lemma}

\begin{proof}
Write the von Neumann entropy $S(\rho)=-\Tr\rho\log\rho$. Then
\begin{align}
\Delta_{T,\sigma}(\rho)
 &=S(T_*\rho)-S(\rho)
   +\Tr\rho\bigl[T(\log\sigma)-\log\sigma\bigr].
\label{eq:deficit-convexity}
\end{align}
By Stinespring's representation
\cite[Theorem~1]{Stinespring1955}, there is an isometry
$U:\cH\to\cH\otimes\cH_F$ such that
\begin{align*}
T_*(\rho)&=\Tr_{\cH_F}U\rho U^*=\omega_{\cH}\ ,\qquad
\omega=U\rho U^*,\\
S(T_*\rho)-S(\rho)&=S(\omega_{\cH})-S(\omega)=:-S(F|{\cH})_\omega.
\end{align*}
Since the last term is linear \eqref{eq:deficit-convexity} and by the convexity of negative
conditional entropy, $\rho \mapsto \Delta_{T,\sigma}(\rho)$ is
convex. The final inequality follows from
$\Delta_{T,\sigma}(\sigma)=0$.
\end{proof}
\section{Entropy contraction for quantum channels}
\label{sec:channels}

Throughout this section, let $E$ be a faithful conditional expectation
$E:\cB(\cH)\to\cN$ and $\Phi:\cB(\cH)\to\cB(\cH)$ be a completely positive unital map satisfying the assumption \eqref{eq:channel-compatibility},
\begin{align*}
\Phi E=E\Phi=E.
\end{align*}
The first lemma concerns the faithful conditional expectation $E$ alone.

\subsection{Conditional variance and reverse relative entropy}

Let $\rho$ be a faithful state, and set $\sigma=E_*\rho$.
For $Z_u=(\mathsf L_\rho+u\mathsf R_\sigma)^{-1}(\rho-\sigma)$
as in \eqref{eq:resolvent-definitions}, we introduce the following auxiliary quantity:
\begin{align}
\mathcal S_E(\rho)&:=\int_0^\infty \frac{u}{(1+u)^2}
  \bigl(\Var_{E,\sigma}(Z_u^*)+u\Var_{E,\sigma}(Z_u)\bigr)\dd u.
\label{eq:conditional-variance-resolvent}
\end{align}Thus $\mathcal S_E(\rho)$ denotes the integral of conditional variances
contributed by the centered operators $Z_u-E(Z_u)$.
Since $E_*\sigma=\sigma$, the orthogonality
\eqref{eq:orthogonal-decompositions} implies
\begin{align}
\Tr(\sigma X^*X)\ge \Var_{E,\sigma}(X)\ ,\qquad
\Tr(\sigma XX^*)\ge \Var_{E,\sigma}(X^*).
\end{align}
The representation of reverse relative entropy
\eqref{eq:reverse-energy} gives
\begin{align*}
0&\le\mathcal S_E(\rho)\le D(\sigma\Vert\rho).
\end{align*}
The following lemma gives a lower bound for
$\mathcal S_E(\rho)$ under an order assumption on $\rho$.

\begin{lemma}
\label{lem:centering}
Let $E:\mathcal B(\mathcal H)\to\mathcal N$ be a faithful
conditional expectation, and let $\rho$ be a faithful state.
If $\rho\ge aE_*\rho$ for some $0<a<1$, then
\[
\mathcal S_E(\rho)\ge aD(E_*\rho\Vert\rho).
\]
\end{lemma}
\begin{proof}
\textbf{Step 1. Integral representation.}\par\noindent
Write $\sigma=E_*\rho$ and set
\[
\mathcal M_E(\rho)
:=D(\sigma\Vert\rho)-\mathcal S_E(\rho).
\]
It suffices to prove
$\mathcal M_E(\rho)\le(1-a)D(\sigma\Vert\rho)$.
We work with the GNS inner product
\[
\langle X,Y\rangle_\sigma
=
\operatorname{Tr}(\sigma X^*Y),
\]
and define the modular operators
\[
\Delta_{\rho,\sigma}=L_\rho R_\sigma^{-1},
\qquad
\Delta_{\sigma}=L_\sigma R_\sigma^{-1}.
\]
These operators are positive definite and self-adjoint for the $L_2(\sigma)$
inner product. Indeed,
\[
\langle X,\Delta_{\rho,\sigma}Y\rangle_\sigma
=
\operatorname{Tr}(X^*\rho Y),
\qquad
\langle X,\Delta_{\sigma}Y\rangle_\sigma
=
\operatorname{Tr}(X^*\sigma Y).
\]
The assumption $\rho\ge a\sigma$ therefore gives
\[
\Delta_{\rho,\sigma}\ge a\Delta_{\sigma}.
\]
Since $E_*\sigma=\sigma$, the conditional expectation $E$ is also the
orthogonal projection onto $\mathcal N\subset L_2(\sigma)$.
For $A\in\mathcal N$ and $X\in\mathcal B(\mathcal H)$,
\[
\begin{aligned}
\langle A,\Delta_\sigma X\rangle_\sigma
&=\operatorname{Tr}(\sigma XA^*)
=\operatorname{Tr}(\sigma E(XA^*))\\
&=\operatorname{Tr}(\sigma E(X)A^*)
=\langle A,\Delta_\sigma E(X)\rangle_\sigma.
\end{aligned}
\]
Thus $E\Delta_\sigma=E\Delta_\sigma E$; taking adjoints gives
$\Delta_\sigma E=E\Delta_\sigma E$.
Write
\[
\Delta_\cN=\Delta_\sigma|_{\mathcal N}.
\]
For $A,B\in\mathcal N$,
\[
\begin{aligned}
\langle A,\Delta_{\rho,\sigma}B\rangle_\sigma
&=\operatorname{Tr}(\rho BA^*)
=\operatorname{Tr}((E_*\rho)BA^*)\\
&=\langle A,\Delta_\sigma B\rangle_\sigma.
\end{aligned}
\]
Thus,
\[
E\Delta_{\rho,\sigma}|_{\mathcal N}=\Delta_\cN.
\]
See also \cite[Lemma~4.17]{GaoWilde2021}.

For $u>0$, define operators on the Hilbert subspace
$\mathcal N\subset L_2(\sigma)$ by
\[
K_u
=
E(\Delta_{\rho,\sigma}+u\,\mathrm{id})^{-1}|_{\mathcal N}
-
(\Delta_\cN+u\,\mathrm{id})^{-1},
\qquad
f_u=K_u(I).
\]
Operator Jensen gives $K_u\ge0$.
Since $E$ commutes with $\Delta_\sigma$, the inequality
$\Delta_{\rho,\sigma}\ge a\Delta_\sigma$ and the
order-reversing property of inversion yield
\[
0\le K_u\le F_u,
\qquad
F_u
=
(a\Delta_\cN+u\,\mathrm{id})^{-1}
-
(\Delta_\cN+u\,\mathrm{id})^{-1}.
\]
In particular, $K_u$ is self-adjoint and $F_u$ is positive definite
with respect to the Hilbert space $L_2(\cN,\sigma)$.

Consider the equation
\[
\rho-\sigma
=\rho Z_u+uZ_u\sigma
=R_{\sigma}(\Delta_{\rho,\sigma}+u\,\mathrm{id})(Z_u).
\]
The solution can be expressed as
\[
Z_u
=
I-(1+u)(\Delta_{\rho,\sigma}+u\,\mathrm{id})^{-1}(I).
\]
Since $\Delta_\sigma(I)=I$, we obtain
\begin{align*}
E(Z_u)
&=I-(1+u)\left[
K_u(I)+(\Delta_\cN+u\,\mathrm{id})^{-1}(I)
  \right]\\
&=-(1+u)f_u.
\end{align*}
Taking the trace in the defining equation for $Z_u$ gives
\begin{align*}
\mathcal Q_u
&=\operatorname{Tr}((\rho-\sigma)Z_u)
=-(1+u)\operatorname{Tr}(\sigma Z_u)\\
&=-(1+u)\operatorname{Tr}(\sigma E(Z_u))
=(1+u)^2\operatorname{Tr}(\sigma f_u).
\end{align*}
The resolvent representation of reverse relative entropy therefore
implies
\[
D(\sigma\Vert\rho)
=\int_0^\infty\frac{\mathcal Q_u}{(1+u)^2}\,du
=\int_0^\infty\langle I,f_u\rangle_\sigma\,du.
\]
On the other hand, the orthogonal decomposition associated with $E$
and the energy representation of reverse relative entropy give
\[
\begin{aligned}
\mathcal M_E(\rho)
&=\int_0^\infty\frac{u}{(1+u)^2}
\left(
\operatorname{Tr}\!\left[\sigma E(Z_u)E(Z_u)^*\right]
+u\operatorname{Tr}\!\left[\sigma E(Z_u)^*E(Z_u)\right]
\right)\,du\\
&=\int_0^\infty
\left(
u\langle f_u,\Delta_\cN f_u\rangle_\sigma
+u^2\|f_u\|_\sigma^2
\right)\,du.
\end{aligned}
\]
Here we used
$
\langle X,\Delta_\cN X\rangle_\sigma
=
\operatorname{Tr}(\sigma XX^*)
=
\|X^*\|_\sigma^2,
\ X\in\mathcal N.
$

\noindent\textbf{Step 2. Integrated modular symmetry.}\par\noindent
The following integrated symmetry will allow us to compare
the quadratic terms in $\mathcal M_E(\rho)$ with those
obtained from $K_u\le F_u$.
For $p\in\{0,1,2\}$ and $q\in\mathbb R$, define
\[
J_{p,q}
=
\int_0^\infty
u^p\langle f_u,\Delta_\cN^qf_u\rangle_\sigma\,du.
\]
We claim that
\[
J_{p,q}=J_{p,\,2-p-q}.
\]
All integrals below converge. Indeed, $K_u$ is bounded near zero,
and the resolvent expansions, together with
$E\Delta_{\rho,\sigma}|_{\mathcal N}=\Delta_\cN$, imply
\[
K_u=O(u^{-3})\quad\text{as }u\to\infty.
\]

For $s>0$, let $\Pi_s$ be the spectral projection of
$\Delta_\cN=\Delta_\sigma|_{\mathcal N}$ corresponding to
the eigenvalue $s$. If $A\in\mathcal N$ satisfies
$\Delta_\cN A=sA$, then
\[
\sigma A\sigma^{-1}=sA,
\qquad
\Delta_{\rho,\sigma}R_A=sR_A\Delta_{\rho,\sigma},
\]
where $R_A(X)=XA$ is the right multiplication. Thus
\[
(\Delta_{\rho,\sigma}+u\,\mathrm{id})^{-1}R_A
=
s^{-1}R_A
(\Delta_{\rho,\sigma}+(u/s)\,\mathrm{id})^{-1}.
\]
Using the right $\mathcal N$-module property of $E$, we obtain
\[
E(\Delta_{\rho,\sigma}+u\,\mathrm{id})^{-1}(A)
=s^{-1}\left(f_{u/s}+\frac{1}{1+u/s}I\right)A.
\]
Subtracting
$(\Delta_\cN+u\,\mathrm{id})^{-1}(A)=(s+u)^{-1}A$
gives
\[
K_u(A)=s^{-1}f_{u/s}A.
\]
Self-adjointness of $K_u$ now gives
\[
\begin{aligned}
\langle A,f_u\rangle_\sigma
&=\langle K_uA,I\rangle_\sigma
=s^{-1}\langle f_{u/s}A,I\rangle_\sigma\\
&=s^{-1}\langle A,f_{u/s}^*\rangle_\sigma.
\end{aligned}
\]
Since taking adjoints sends the $s$-eigenspace of
$\Delta_\cN=L_\sigma R_\sigma^{-1}$ onto its
$s^{-1}$-eigenspace, it follows that
\[
\Pi_sf_u
=
s^{-1}\bigl(\Pi_{1/s}f_{u/s}\bigr)^*.
\]
Moreover, if $\Delta_\cN X=sX$, then
\[
\|X^*\|_\sigma^2=s\|X\|_\sigma^2.
\]
Therefore, writing
\[
n_s(u)=\|\Pi_sf_u\|_\sigma^2,
\]
we have
\[
n_s(u)=s^{-3}n_{1/s}(u/s).
\]
Changing variables $u=sv$ in each spectral summand gives
\[
\begin{aligned}
J_{p,q}
&=\sum_s s^q\int_0^\infty u^p n_s(u)\,du\\
&=\sum_s s^{p+q-2}\int_0^\infty v^p n_{1/s}(v)\,dv\\
&=J_{p,\,2-p-q},
\end{aligned}
\]
which proves the claim.

\noindent\textbf{Step 3. Closing by operator comparison.}\par\noindent
It follows from the above symmetry
\[
J_{1,1}=J_{1,0},
\qquad
J_{2,-1}=J_{2,1}.
\]
Since $\Delta_\cN+\Delta_\cN^{-1}\ge2\,\mathrm{id}$, we obtain
\[
\begin{aligned}
J_{2,-1}
&=\frac12\bigl(J_{2,-1}+J_{2,1}\bigr)\\
&=\frac12\int_0^\infty
u^2\langle f_u,(\Delta_\cN+\Delta_\cN^{-1})f_u\rangle_\sigma\,du\\
&\ge J_{2,0}.
\end{aligned}
\]
Since
\[
F_u
=
(1-a)\Delta_\cN
(a\Delta_\cN+u\,\mathrm{id})^{-1}
(\Delta_\cN+u\,\mathrm{id})^{-1},
\]
we have
\[
(1-a)F_u^{-1}
=
a\Delta_\cN+(1+a)u\,\mathrm{id}+u^2\Delta_\cN^{-1}.
\]
Since $0\le K_u\le F_u$, squaring the positive contraction
$F_u^{-1/2}K_uF_u^{-1/2}$ gives
\[
K_uF_u^{-1}K_u\le K_u.
\]
Evaluating this inequality at $I$ and using $f_u=K_u(I)$ gives
\[
\begin{aligned}
(1-a)\langle I,f_u\rangle_\sigma=(1-a)\langle I,K_u(I)\rangle_\sigma
&\ge(1-a)\langle I,K_uF_u^{-1}K_u(I)\rangle_\sigma\\
&=(1-a)\langle f_u,F_u^{-1}f_u\rangle_\sigma\\
&=a\langle f_u,\Delta_\cN f_u\rangle_\sigma
 +(1+a)u\|f_u\|_\sigma^2
 +u^2\langle f_u,\Delta_\cN^{-1}f_u\rangle_\sigma.
\end{aligned}
\]
After integration, the modular symmetry yields
\[
\begin{aligned}
(1-a)D(\sigma\Vert\rho)
&\ge aJ_{0,1}+(1+a)J_{1,0}+J_{2,-1}\\
&\ge J_{1,1}+J_{2,0}\\
&=\mathcal M_E(\rho).
\end{aligned}
\]
Therefore,
\begin{align*}
\mathcal S_E(\rho)
=&
D(\sigma\Vert\rho)-\mathcal M_E(\rho)
\ge aD(\sigma\Vert\rho). \qedhere
\end{align*}
\end{proof}

\subsection{Relative entropy loss under a channel}
Let $\Phi$ be unital and completely positive, and let
$\rho,\sigma,\Phi_*\rho,\Phi_*\sigma$ be faithful states.
Recall the notations $Z_u$ and $\mathcal Q_u$ from \eqref{eq:resolvent-definitions}
\begin{align*}
Z_u&=(\mathsf L_\rho +u\mathsf R_\sigma)^{-1}(\rho-\sigma),\\
\mathcal Q_u&=\ip{\rho-\sigma}{(\mathsf L_\rho +u\mathsf R_\sigma)^{-1}(\rho-\sigma)}_{\HS}=\Tr(\rho Z_uZ_u^*)+u\Tr(\sigma Z_u^*Z_u).
\end{align*}
for the input pair $\rho,\sigma$ and denote $Z'_u,\mathcal Q'_u$ for the corresponding terms for the
output pair $(\Phi_*\rho,\Phi_*\sigma)$. Recall the inverse-form
monotonicity of Lesniewski and Ruskai
\cite[Theorem~2.15 and Eqs.~(58)--(60)]{LesniewskiRuskai1999} that
\begin{align*}
\ip{\Phi_*(A)}{(\mathsf L_{\Phi_*\rho}+u\mathsf R_{\Phi_*\sigma})^{-1}\Phi_*(A)}_{\HS}
 \le \ip{A}{(\mathsf L_\rho+u\mathsf R_\sigma)^{-1}A}_{\HS}, \ \ \forall \ A\in B(\cH)
 \qquad u>0,
\end{align*}
Taking $A=\rho-\sigma$ gives $\mathcal Q'_u\le\mathcal Q_u$. The following identity retains this
monotonicity in a completed-square form. Define the mixed GNS-norm
\begin{align*}
\norm{Y}_{\rho,\sigma;u}^2
 &=\Tr(\rho YY^*)+u\Tr(\sigma Y^*Y),
\end{align*} and write \[\mathcal B_\Phi(X)=\Phi(X^*X)-\Phi(X)^*\Phi(X)\ge0,\] Kadison-Schwarz inequality. It is (essentially) shown in  \cite[Section~4.3, Eqs.~(34)--(36)]{Ruskai2007} that
\begin{align}
\mathcal{Q}_u-\mathcal{Q}'_u
 =\norm{Z_u-\Phi(Z'_u)}_{\rho,\sigma;u}^2
   +\Tr\bigl[\rho\mathcal{B}_\Phi((Z'_u)^*)\bigr]
   +u\Tr\bigl[\sigma\mathcal{B}_\Phi(Z'_u)\bigr].
\label{eq:finite-resolvent-defect}
\end{align}
Indeed, set $A=\rho-\sigma$, and
\begin{align*}
F_u(Y)&=2\Re\ip{A}{Y}_{\HS}-\norm{Y}_{\rho,\sigma;u}^2,
\end{align*}
with $F'_u$ defined for the output pair $\Phi_*\rho,\Phi_*\sigma$. The variational formula
\eqref{eq:variational} gives
\begin{align*}
\mathcal Q_u-F_u(\Phi Z'_u)
 &=\norm{Z_u-\Phi Z'_u}_{\rho,\sigma;u}^2,
 \qquad F'_u(Z'_u)=\mathcal Q'_u.
\end{align*}Therefore
\begin{align}
\mathcal Q_u-\mathcal Q'_u=\norm{Z_u-\Phi Z'_u}_{\rho,\sigma;u}^2+F_u(\Phi Z'_u)-F'_u(Z'_u)\label{eq:resolvent-defect-splitting}
\end{align}
The linear terms cancel by trace duality
\[ \ip{A}{\Phi Z'_u}_{\HS}=\ip{\Phi_*A}{Z'_u}_{\HS}\ .\]
The remaining terms are
\begin{align*}
&F_u(\Phi Z'_u)-F'_u(Z'_u)\\
 =&\Tr[(\Phi_*\rho)Z'_u(Z'_u)^*]+u\Tr[(\Phi_*\sigma)(Z'_u)^*Z'_u]\\
 &\quad-\Tr[\rho\Phi(Z'_u)\Phi(Z'_u)^*]-u\Tr[\sigma\Phi(Z'_u)^*\Phi(Z'_u)]\\
 =&\Tr\bigl[\rho\bigl(\Phi(Z'_u(Z'_u)^*)
                         -\Phi(Z'_u)\Phi(Z'_u)^*\bigr)\bigr]+u\Tr\bigl[\sigma\bigl(\Phi((Z'_u)^*Z'_u)
                         -\Phi(Z'_u)^*\Phi(Z'_u)\bigr)\bigr]\\
 =&\Tr[\rho\mathcal B_\Phi((Z'_u)^*)]
       +u\Tr[\sigma\mathcal B_\Phi(Z'_u)].
\end{align*}
Putting this back to \eqref{eq:resolvent-defect-splitting} proves \eqref{eq:finite-resolvent-defect}.

\begin{lemma}
\label{lem:channel-defect}
Let $\Phi$ be unital and completely positive and $E$ be a faithful conditional expectation
satisfying $\Phi E=E\Phi=E$.
If for a state $\rho$, $\rho\ge a E_*\rho$
for some $0<a<1$,  then
\begin{align}
D(\rho\Vert E_*\rho)-D(\Phi_*\rho\Vert E_*\rho)
 &\ge a^2\kappa D(E_*\rho \Vert \Phi_*\rho).
\end{align}
where $\kappa$ is the optimal constant from \eqref{eq:channel-variance} such that
\begin{align}
\Tr\bigl[\omega \mathcal B_\Phi(X)\bigr]
 &\ge\kappa\Var_{E, \omega}(X), \qquad  \forall \  X
\end{align}
for any faithful invariant $\omega=E_*\omega$.
\end{lemma}

\begin{proof}
Take $\sigma=E_*\rho$. We have
\begin{align*}
\Phi_*\sigma&=\sigma,\qquad E_*\Phi_*\rho=\sigma,\qquad
\Phi_*\rho\ge a\Phi_*\sigma=a\sigma.
\end{align*}
Drop the square, use positivity of $\mathcal B_\Phi$ and $\rho\ge a\sigma$,
and the definition of $\kappa$ above
\begin{align*}
\mathcal Q_u-\mathcal Q'_u
 &\ge\Tr[\rho\mathcal B_\Phi((Z'_u)^*)]
             +u\Tr[\sigma\mathcal B_\Phi(Z'_u)]\\
 &\ge a\Bigl(\Tr[\sigma\mathcal B_\Phi((Z'_u)^*)]
             +u\Tr[\sigma\mathcal B_\Phi(Z'_u)]\Bigr)\\
 &\ge a\kappa\Bigl(\Var_{E,\sigma}((Z'_u)^*)
                         +u\Var_{E,\sigma}(Z'_u)\Bigr).
\end{align*}
The forward representation \eqref{eq:forward-resolvent}, definition
\eqref{eq:conditional-variance-resolvent}, and Lemma~\ref{lem:centering}
applied to $\Phi_*\rho$ yield
\begin{align*}
D(\rho\Vert\sigma)-D(\Phi_*\rho\Vert\sigma)
 &=\int_0^\infty\frac{u}{(1+u)^2}(\mathcal Q_u-\mathcal Q'_u)\dd u\\
 &\ge a\kappa\mathcal S_E(\Phi_*\rho)
 \ge a^2\kappa D(\sigma\Vert\Phi_*\rho).\qedhere
\end{align*}
\end{proof}

\subsection{The logarithmic index bound}

\begin{proof}[Proof of Theorem~\ref{thm:channel}]
Let $C=C(E)$ and $\rho$ be a faithful state. Put
\begin{align*}
\sigma&=E_*\rho=E_*\Phi_*\rho,\qquad
\widehat\rho=\frac{\rho+\sigma}{2},\qquad
\widehat{\Phi_*\rho}=\Phi_*\widehat\rho=\frac{\Phi_*\rho+\sigma}{2},\\
\Delta(\rho)&=D(\rho\Vert\sigma)-D(\Phi_*\rho\Vert\sigma).
\end{align*}
We have the order relations
\begin{align*}
E_*\widehat\rho&=E_*\widehat{\Phi_*\rho} =\sigma,\qquad
\tfrac12\sigma\le\widehat{\Phi_*\rho}\le \frac{C+1}{2}\sigma.
\end{align*}
The upper bound follows from $\Phi_*\rho\le C\sigma$ by averaging with $\sigma$.
By the convexity from Lemma~\ref{lem:ep-convexity},
\[
\Delta(\widehat\rho)\le\frac{1}{2}(\Delta(\sigma)+\Delta(\rho))=\frac{1}{2}\Delta(\rho).
\]
Using Lemma~\ref{lem:channel-defect} with $a=1/2$ and then the two
comparisons of Lemma~\ref{lem:entropy-comparisons} for the output,
\begin{align}
D(\rho\Vert\sigma)
-D(\Phi_*\rho\Vert\sigma)=\Delta(\rho)
 &\ge 2\Delta(\widehat\rho)
 \ge\frac{\kappa}{2}D(\sigma\Vert\widehat{\Phi_*\rho})\nonumber\\
 &\ge\frac{\kappa}{2(1+\log C)}D(\widehat{\Phi_*\rho}\Vert\sigma)\nonumber\\
 &\ge\frac{(1-\log2)\kappa}{4(1+\log C)}D(\Phi_*\rho\Vert\sigma).
\end{align}
Rearranging the term gives,
\begin{align}
D(\rho\Vert\sigma)
 &=\Delta(\rho)+D(\Phi_*\rho\Vert\sigma)\nonumber\\
 &\ge\left(1+\frac{c_0\kappa}{1+\log C}\right)
D(\Phi_*\rho\Vert\sigma).
\end{align}
where $c_0= (1-\log 2)/4$. The case of arbitrary states follows from perturbation
$\rho_\eps=(1-\eps)\rho+\eps\omega$.

Applying the result to matrix amplification $\Phi^{(n)}=\Phi\otimes \id_{M_n}$ with $C(E^{(n)})\le
C_{\cb}(E)$ and the amplification invariance of $\kappa$ from
Proposition~\ref{prop:channel-stability} gives the complete entropy contraction.
\end{proof}

\section{MLSI for quantum Markov semigroups}
\label{sec:semigroups}
\subsection{Modified logarithmic Sobolev inequalities}

A QMS $P_t=e^{t\cL}:\cB(\cH)\to \cB(\cH)$ is a norm-continuous semigroup
of unital completely positive maps. Its generator $\cL$ has
the Gorini--Kossakowski--Lindblad--Sudarshan (GKLS) form
\cite[Theorem~2]{Lindblad1976}
\begin{align}
\cL(X)=i[H,X]+\sum_j\cD_{V_j}(X),\qquad
\cD_V(X)=V^*XV-\tfrac12\{V^*V,X\}.
\end{align}
where $H=H^*$ and $V\in \cB(\cH)$.  We assume that $P_t$ admits a faithful invariant state
$\sigma=P_{t*}(\sigma)$ and
the time limit \[
E=\displaystyle \lim_{t\to \infty }P_t
\]
exists. Thus $E=E\circ E$ a faithful conditional
expectation onto the fixed point space
\begin{align}
\cN=\Ran E=\Ker\cL=\{X\ |\ P_t(X)=X, \ \forall\  t\ge 0\}
\end{align}
Consequently,
\begin{align}
P_tE=EP_t=E,\qquad
P_{t*}E_*=E_*P_{t*}=E_*,\qquad
E_*\omega=\omega\ \Longleftrightarrow\ \cL_*\omega=0,
\end{align}
where $P_{t*}=e^{t\cL_*}$ is the trace adjoint as the state evolution, and $E_*\rho$ is its
limit equilibrium state.
The case $\cN=\mathbb C\Id$ is called primitive;
then $E(X)=\Tr(\sigma X)\Id$ is the completely depolarizing map for the unique invariant state $\sigma$.

For a faithful state $\rho$, the entropy production is
\cite{Spohn1978}
\begin{align}
\EP_{\cL}(\rho)
&=-\left.\frac{\dd}{\dd t}\right|_{t=0}
D(P_{t*}\rho\Vert E_*\rho)
=-\Tr\bigl[\cL_*(\rho)(\log\rho-\log E_*\rho)\bigr].
\end{align}
Data processing makes this quantity nonnegative. The exponential
decay estimate \eqref{eq:intro-decay} is equivalent to the modified log-Sobolev inequality
\begin{align}
2\alpha D(\rho\Vert E_*\rho)\le \EP_{\cL}(\rho),
\end{align}
for faithful states.
We define the optimal MLSI constant $\alpha$ and its complete version $\alpha^{\mathrm c}$ as
\begin{align}
\alpha(\cL)
&=\inf_{\substack{\rho>0,\ \Tr\rho=1\\\rho\ne E_*\rho}}
\frac{\EP_{\cL}(\rho)}{2D(\rho\Vert E_*\rho)},\qquad
\alpha^{\mathrm c}(\cL)=\inf_{n\ge1}\alpha(\cL\otimes \id_{M_n}).
\end{align}
For a faithful invariant state $\omega$, the GNS adjoint $\cL^\sharp$ and
symmetrization $\cL_{\GNS}$ are defined by
\begin{align}
\ip{X}{\cL(Y)}_{2,\omega}
&=\ip{\cL^\sharp(X)}{Y}_{2,\omega},\qquad
\cL_{\GNS}=\tfrac12(\cL+\cL^\sharp).
\end{align}
Set $Q_\omega(X)=-\Re\ip{X}{\cL(X)}_{2,\omega}$. The GNS gap is
\begin{align}
\lambda_{\GNS}(\cL)
&=\inf_{\substack{X\ne0\\EX=0}}
\frac{Q_\omega(X)}{\norm{X}_{2,\omega}^2}
=\min\spec\!\left(-\cL_{\GNS}\big|_{\Ker E}\right).
\label{eq:gns-gap}
\end{align}
The infimum is over all complex matrices $X$ such that $X\neq 0$ and $EX=0$. It can vanish even when
$\Ker\cL=\cN$ \cite[Section~6]{GaoGuo2026}.

\Needspace{8\baselineskip}
\begin{proposition}[GNS coercivity]

The adjoint $\cL^\sharp$ and the GNS gap $\lambda_{\GNS}(\cL)\ge0$ are independent of the choice of faithful
invariant state $\omega$. Moreover, $\lambda_{\GNS}(\cL)=\lambda_{\GNS}(\cL\otimes \id_{M_n}), n\ge
1$ is unchanged by matrix amplification.
For every faithful invariant state $\omega$ and every $X$,
\begin{align}
Q_\omega(X)
&=\frac12\sum_j\Tr\bigl(\omega[V_j,X]^*[V_j,X]\bigr)
\ge \lambda_{\GNS}(\cL) \Var_{E,\omega}(X).
\end{align}
\end{proposition}
\begin{proof}
Gao and Guo \cite[Proposition~5.1]{GaoGuo2026} show that
$\cL^\sharp$ is independent of $\omega$ and that, on $L_2(\omega)$,
\begin{align*}
-\cL_{\GNS}\ge0,\qquad
E\cL=\cL E=E\cL^\sharp=\cL^\sharp E=0.
\end{align*}
Moreover, \cite[Proposition~5.3]{GaoGuo2026} gives
\begin{align}
\lambda_{\GNS}(\cL^{(n)})&=\lambda_{\GNS}(\cL).
\end{align}
for all $n\ge 1$. In particular, $\lambda_{\GNS}(\cL)$ is a nonnegative spectral quantity
independent of $\omega$.
Note that these are Hilbert-space assertions; $\cL_{\GNS}$ need not generate a QMS.
The gradient form is
\begin{align*}
2\Gamma(X,X):&=\cL(X^*X)-\cL(X^*)X-X^*\cL(X)\\
&=\sum_j\bigl(V_j^*X^*XV_j-V_j^*X^*V_jX-X^*V_j^*XV_j+X^*V_j^*V_jX\bigr)\\
&=\sum_j[V_j,X]^*[V_j,X].
\end{align*}
Put $Y=X-EX$. Invariance $\cL_*\omega=0$ and GNS orthogonality give
\begin{align*}
\Tr[\omega\cL(X^*X)]&=0,\qquad
\Tr[\omega\cL(X^*)X]
=\overline{\Tr[\omega X^*\cL(X)]},\\
EY&=0,\qquad \cL(X)=\cL(Y),\qquad
\ip{EX}{\cL(Y)}_{2,\omega}
=\ip{EX}{E\cL(Y)}_{2,\omega}=0.
\end{align*}
Taking the invariant state for the gradient form identity, and
then applying \eqref{eq:gns-gap} yields
\begin{align*}
\frac12\sum_j\Tr\bigl(\omega[V_j,X]^*[V_j,X]\bigr)&=-\Re\ip{X}{\cL(X)}_{2,\omega}
=-\Re\ip{Y}{\cL(Y)}_{2,\omega}\\
&=\ip{Y}{-\cL_{\GNS}(Y)}_{2,\omega}\\
&\ge\lambda_{\GNS}(\cL)\norm{Y}_{2,\omega}^2=\lambda_{\GNS}(\cL)\Var_{E,\omega}(X). \qedhere
\end{align*}
\end{proof}
\subsection{Proof of Corollary~\ref{cor:main}}

We first record a finite-time consequence of Theorem~\ref{thm:channel}.
\begin{corollary}[Strict contraction at every positive time]

Let $P_t=e^{t\cL}$ have a faithful asymptotic conditional expectation
$E$. For a faithful invariant state $\omega$, put
\begin{align*}
s(t)&=\norm{P_t-E}_{2,\omega\to2,\omega},\qquad
\kappa(t)=1-s(t)^2.
\end{align*}
Then $\kappa(t)>0$ for every $t>0$, and
\begin{align*}
\eta(P_t\mid E)
&\le\left(1+\frac{c_0\kappa(t)}{1+\log C(E)}\right)^{-1}<1.
\end{align*}
The complete bound holds with $\eta^{\mathrm c}$ and $C_{\cb}(E)$.
These bounds do not require a positive GNS gap.
\end{corollary}

\begin{proof}
For $X\in\Ker E$, Schwarz and invariance give
\begin{align*}
\norm{P_{t+h}X}_{2,\omega}^2
&\le\Tr\bigl[\omega P_h((P_tX)^*P_tX)\bigr]
=\norm{P_tX}_{2,\omega}^2,
\qquad t,h\ge0.
\end{align*}
If $\norm{X}_{2,\omega}=\norm{P_hX}_{2,\omega}=1$ for some $h>0$,
monotonicity yields
\begin{align*}
\norm{P_tX}_{2,\omega}^2=1\quad(0\le t\le h)
\quad\Longrightarrow\quad
\norm{P_tX}_{2,\omega}^2=1\quad(t\ge0).
\end{align*}
The implication uses real analyticity of the matrix exponential.
It contradicts $P_tX\to EX=0$. Compactness of the unit sphere of
$\Ker E$ now gives $s(h)<1$, and Theorem~\ref{thm:channel} applies.

For any fixed $h>0$, write $t=mh+r$, where $m=\lfloor t/h\rfloor$ and
$0\le r<h$. With $\sigma=E_*\rho$, iteration and data processing give
\begin{align*}
D(P_{t*}\rho\Vert\sigma)
&\le D(P_{mh*}\rho\Vert\sigma)\\
&\le\left(1+\frac{c_0\kappa(h)}{1+\log C(E)}\right)^{-m}
D(\rho\Vert\sigma).
\end{align*}
Thus
\begin{align}
\eta(P_t\mid E)
&\le\left(1+\frac{c_0\kappa(h)}{1+\log C(E)}\right)^{-\lfloor t/h\rfloor}.
\end{align}
Applying the complete channel bound at each step proves the same
formula for $\eta^{\mathrm c}$ with $C_{\cb}(E)$.
\end{proof}

\begin{proof}[Proof of Corollary~\ref{cor:main}]
Let $P_t=e^{t\cL}$ be QMS with a faithful asymptotic conditional expectation
$E$. Fix any faithful invariant state $\omega$. Set
\begin{align*}
s(t)&=\norm{P_t-E}_{2,\omega\to2,\omega},\qquad
\kappa(t)=1-s(t)^2.
\end{align*}
Note that for any $X\in B(\cH)$,
\begin{align*} \norm{(P_t-E)(X)}_{2,\omega}^2=& \langle (P_t-E)(X), (P_t-E)(X)\rangle_{\omega}\\
=& \langle (X), (P_t^\sharp-E)(P_t-E)(X)\rangle_{\omega}\\
=& \langle (X), (P_t^\sharp P_t-E)(X)\rangle_{\omega}
\end{align*}
where $P_t^\sharp$ is the GNS-adjoint of $P_t$ with respect to $L_2(\omega)$ and we used $P_t^\sharp
E=EP_t=E$. Here $P_h^\sharp P_h$ is a self-adjoint operator on $L_2(\omega)$ with the expansion,
\begin{align*}
P_h^\sharp P_h
&=e^{h\cL^\sharp}e^{h\cL}
=\id+h(\cL^\sharp+\cL)+O(h^2)\\
&=\id+2h\cL_{\GNS}+O(h^2),
\end{align*}
Note that $E$ is a projection commute with $P_h^\sharp P_h$, and $\text{Ran}(E)$ is the eigenspace
of $P_h^\sharp P_h$ with eigenvalue $1$. By \eqref{eq:gns-gap},
we have on the centered space $\Ker E$,
\begin{align*}
s(h)^2
&=\max\spec\!\left(P_h^\sharp P_h\big|_{\Ker E}\right)
=1-2\lambda_{\GNS}(\cL)h+O(h^2),\\
\kappa(h)&=1-s(h)^2=2\lambda_{\GNS}(\cL)h+O(h^2).
\end{align*}
Put $C=C(E)$ and $\sigma=E_*\rho$. For $t>0$, take $h=t/m$. Apply the channel estimate
Theorem~\ref{thm:channel} for the iterated form $P_{t*}=P_{h*}^m$,
\begin{align*}
D(P_{t*}\rho\Vert E_*\rho )=D(P_{\frac{t}{m}*}^m\rho\Vert E_*\rho )
\le&\left(1+\frac{c_0\kappa(t/m)}{1+\log C}\right)^{-m}
D(\rho\Vert\sigma),\\
m\log\left(1+\frac{c_0\kappa(t/m)}{1+\log C}\right)
=&\frac{2c_0\lambda_{\GNS}(\cL)t}{1+\log C}+O(m^{-1}).
\end{align*}
Letting $m\to\infty$ gives
\begin{align*}
D(P_{t*}\rho\Vert\sigma)
&\le\exp\!\left(-\frac{2c_0\lambda_{\GNS}(\cL)t}{1+\log C}\right)
D(\rho\Vert\sigma),\\
\alpha(\cL)&\ge\frac{c_0\lambda_{\GNS}(\cL)}{1+\log C}.
\end{align*}
Applying the estimate for all matrix amplification $P_t^{(n)}=P_t\otimes \id_{M_n}$ with the
complete index $C_{\cb}(E)$ yields the complete entropy decay.
\end{proof}

\begin{remark}
For GNS-symmetric
quantum Markov semigroups, Gao, Junge, LaRacuente, and Li
obtained the following complete MLSI bound with logarithmic index
dependence
\cite[Corollary~4.13]{GaoJungeLaRacuenteLi2025}
\begin{align*}
\alpha^{\mathrm c}(\cL)
 &\ge \frac{\lambda_{\GNS}(\cL)}
 {2\log(10C_{\cb}(E))}.
\end{align*}
This bound is stronger than that of Corollary~\ref{cor:main}
when GNS symmetry holds. Indeed, the constant
\[ c_0=(1-\ln 2)/4\approx 0.07<0.5\]
Their tracial result
\cite[Theorem~2.5]{GaoJungeLaRacuenteLi2025} also gives
complete entropy contraction for non-symmetric bistochastic channels.
Our Theorem~\ref{thm:channel} and Corollary~\ref{cor:main} retain
the same logarithmic index dependence for general faithful invariant
states without assuming detailed balance.

\end{remark}

\section{Applications and examples}
\label{sec:applications}

We apply our estimates to alternating pulses and continuous coupling
on a $d$-level system, $d\ge3$. Let $\{|j\rangle\}_{j=0}^{d-1}$
label the levels, and fix downward and upward rates $p>q>0$,
independently of $d$. Each adjacent pair is coupled to an independent
reservoir. For $0\le j\le d-2$, define
\begin{align*}
A_j&=|j\rangle\langle j+1|,\qquad
W_j=\sqrt p A_j+\sqrt q A_j^*,\\
\cD_V(X)&=V^*XV-\tfrac12\{V^*V,X\},
\qquad V,X\in\cB(\mathbb C^d).
\end{align*}
Here $A_j$ lowers level $j+1$ to $j$, $W_j$ combines downward and
upward jumps, and $\cD_V$ is the Heisenberg dissipator.
For squeezing parameters $\theta_j\in[0,1]$, set
\begin{equation}
\begin{aligned}
\mathcal G_j(\theta_j)(X)
 &=\theta_j\cD_{W_j}(X)
 +(1-\theta_j)(p\cD_{A_j}(X)+q\cD_{A_j^*}(X))\\
 &=p\cD_{A_j}(X)+q\cD_{A_j^*}(X)
 +\theta_j\sqrt{pq}(A_j^*XA_j^*+A_jXA_j).
\end{aligned}
\label{eq:local-reservoir-generator}
\end{equation}
At $\theta_j=0$ the jumps act separately; at $\theta_j=1$ they combine
into $W_j$. The cross terms vanish on diagonal observables.
Define the level observable and the equilibrium quantities by
\begin{align}
J_d&=\sum_{j=0}^{d-1}j|j\rangle\langle j|,\qquad
\beta=\log(p/q),\nonumber\\
Z_d&=\Tr(e^{-\beta J_d})=\sum_{j=0}^{d-1}e^{-j\beta},\qquad
\pi_j=Z_d^{-1}e^{-j\beta},\nonumber\\
\sigma_d&=Z_d^{-1}e^{-\beta J_d}
=\sum_{j=0}^{d-1}\pi_j|j\rangle\langle j|,
\qquad E_d(X)=\Tr(\sigma_dX)\Id,
\end{align}
where $\beta>0$ is fixed independently of $d$. The balance relation
$\pi_jq=\pi_{j+1}p$ gives $\mathcal G_{j*}(\theta_j)\sigma_d=0$.
Both applications use the complete index and the relative entropy of
the uniform state:
\begin{align}
C_{\cb}(E_d)&=\Tr(\sigma_d^{-1})=Z_d^2e^{(d-1)\beta},\nonumber\\
D_d&:=D(\Id/d\Vert\sigma_d)
=\tfrac12(d-1)\beta+\log Z_d-\log d.
\label{eq:reservoir-index-entropy}
\end{align}
Since $1\le Z_d\le(1-e^{-\beta})^{-1}$,
$\log C_{\cb}(E_d)\asymp D_d\asymp d$.

We record the action of the local generators on diagonal observables
and coherences. These spaces are invariant and GNS-orthogonal.
Identify diagonal observables with $L_2(\pi)$, with norm
$\|x\|_\pi^2=\sum_i\pi_i|x_i|^2$. The restrictions $Q_j$ of
$\mathcal G_j(\theta_j)$ in \eqref{eq:local-reservoir-generator} satisfy
\begin{align*}
(Q_jx)_j&=q(x_{j+1}-x_j),\qquad
(Q_jx)_{j+1}=p(x_j-x_{j+1}),\\
(Q_jx)_i&=0\quad(i\notin\{j,j+1\}).
\end{align*}
The tridiagonal recurrence gives, for $Q=\sum_jQ_j$,
\begin{align}
\spec(-Q)&=\{0\}\cup
\{p+q-2\sqrt{pq}\cos(k\pi/d):1\le k\le d-1\}.
\label{eq:reservoir-classical-gap}
\end{align}
For $E_{jk}=|j\rangle\langle k|$, $j<k$, each coherence space
$\operatorname{span}\{E_{jk},E_{kj}\}$ is invariant. Writing
$b_i^{(l)}$ for the exit rate from level $i$ due to reservoir $l$,
\begin{align}
b_i^{(l)}&=q\mathbf1_{i=l}+p\mathbf1_{i=l+1},\nonumber\\
\mathcal G_l(\theta_l)E_{jk}
&=-\tfrac12(b_j^{(l)}+b_k^{(l)})E_{jk}
 +\theta_l\sqrt{pq}\,
  \mathbf1_{(j,k)=(l,l+1)}E_{kj}.
\label{eq:reservoir-coherences}
\end{align}

\subsection{Alternating reservoir pulses}

For coupling weights $w_j\ge w_->0$ and pulse duration $\tau>0$, define
\begin{align}
\cL_e&=\sum_{j\ {\rm even}}w_j\mathcal G_j(\theta_j),\qquad
\cL_o=\sum_{j\ {\rm odd}}w_j\mathcal G_j(\theta_j),\nonumber\\
T_e&=e^{\tau\cL_e},\qquad
T_o=e^{\tau\cL_o},\qquad
\Phi_d=T_oT_e.
\end{align}
Both pulses preserve $\sigma_d$, so $\Phi_dE_d=E_d\Phi_d=E_d$.
On states, the odd pulse acts first: $\Phi_{d*}=T_{e*}T_{o*}$.
Each pulse fixes states obtained by reweighting $\sigma_d$ on its
disjoint pairs and unpaired levels, hence
$\eta(T_a\mid E_d)=\eta^{\mathrm c}(T_a\mid E_d)=1$ for $a=e,o$.

Set
\begin{align}
\kappa_*&=\min\left\{
\frac{(1-e^{-2w_-(p+q)\tau})(\sqrt p-\sqrt q)^2}{3(p+q)},\,
1-e^{-w_-q\tau}\right\}>0.
\label{eq:pulse-gap-constant}
\end{align}

\begin{proposition}[Relative entropy contraction for alternating pulses]
With $\kappa_*$ as in \eqref{eq:pulse-gap-constant},
\begin{align}
\frac{c_0\kappa_*}{1+\log C_{\cb}(E_d)+c_0\kappa_*}
&\le1-\eta^{\mathrm c}(\Phi_d\mid E_d)
\le1-\eta(\Phi_d\mid E_d)
\le\frac{2\beta}{D_d}.
\end{align}
Both deficits are of order $d^{-1}$, uniformly in $\theta_j,w_j$,
with constants depending only on $p,q,\tau,w_-$.
\end{proposition}

\begin{proof}
For the lower bound, first consider diagonal observables.
Write $-Q_j=(p+q)P_j$, where $P_j$ are rank-one orthogonal
projections. Projections belonging to disjoint edges have orthogonal
ranges. Set
\begin{align*}
P_e&=\sum_{j\ {\rm even}}P_j,\qquad
P_o=\sum_{j\ {\rm odd}}P_j,\qquad
a_j=1-e^{-w_j(p+q)\tau},\\
T_e&=\id-\sum_{j\ {\rm even}}a_jP_j,\qquad
T_o=\id-\sum_{j\ {\rm odd}}a_jP_j.
\end{align*}
Here the pulse formulas are restricted to diagonal observables.
For $\sum_j\pi_jx_j=0$, let $y=T_ex$ and
$b_-=1-e^{-2w_-(p+q)\tau}$.
Orthogonality and $\|x-y\|_\pi\le\|P_ex\|_\pi$ give
\begin{align*}
\|x\|_\pi^2-\|\Phi_dx\|_\pi^2
&=\sum_{j\ {\rm even}}(2a_j-a_j^2)\|P_jx\|_\pi^2
 +\sum_{j\ {\rm odd}}(2a_j-a_j^2)\|P_jy\|_\pi^2\\
&\ge b_-(\|P_ex\|_\pi^2+\|P_oy\|_\pi^2),\\
\|P_ex\|_\pi^2+\|P_ox\|_\pi^2
&\le\|P_ex\|_\pi^2+2\|P_oy\|_\pi^2+2\|x-y\|_\pi^2\\
&\le3(\|P_ex\|_\pi^2+\|P_oy\|_\pi^2).
\end{align*}
Since $(p+q)(P_e+P_o)=-Q$,
\begin{align*}
\|x\|_\pi^2-\|\Phi_dx\|_\pi^2
&\ge\frac{b_-}{3(p+q)}\langle x,-Qx\rangle_\pi\\
&\ge\frac{b_-(\sqrt p-\sqrt q)^2}{3(p+q)}\|x\|_\pi^2.
\end{align*}

For coherences, use \eqref{eq:reservoir-coherences}.
On adjacent coherences, the pulse containing edge $j$ is a GNS
contraction; the other is $e^{-\tau\ell_j/2}\id$, where
\begin{align*}
\ell_j&=pw_{j-1}+qw_{j+1}\ge w_-q,
\qquad w_{-1}=w_{d-1}=0,\\
\|\Phi_dX\|_{2,\sigma_d}^2
&\le e^{-w_-q\tau}\|X\|_{2,\sigma_d}^2.
\end{align*}
On nonadjacent coherences, both pulses are scalar, and each level has
total exit rate at least $w_-q$, giving the same bound. Hence
\begin{align}
1-\|\Phi_d-E_d\|_{2,\sigma_d\to2,\sigma_d}^2
&\ge\kappa_*.
\end{align}
Thus $\Phi_d^n\to E_d$, and Theorem~\ref{thm:channel} gives the lower bound.

For the upper bound, take $\rho_0=\Id/d$ and $\rho_1=\Phi_{d*}\rho_0$.
The states remain diagonal, and each pulse moves at most one level, so
\begin{align*}
\Tr[J_d(\rho_0-\rho_1)]&\le2,\qquad S(\rho_1)\le\log d,\\
D_d-D(\rho_1\Vert\sigma_d)
&=S(\rho_1)-\log d
 +\beta\Tr[J_d(\rho_0-\rho_1)]\le2\beta.
\end{align*}
Testing on $\rho_0$ proves the upper bound. The order follows from
\eqref{eq:reservoir-index-entropy}.
\end{proof}

\begin{remark}[Symmetry and QMS embedding]

Even at $\theta_j=0$, $\Pr_{\Phi_{d*}}(0\to2)=0$ but
$\Pr_{\Phi_{d*}}(2\to0)>0$, violating GNS symmetry. Thus
\cite[Theorem~4.8]{GaoJungeLaRacuenteLi2025} does not apply directly.
Also,
\begin{align*}
\Phi_{d*}(E_{00})&=(1-u)E_{00}+uE_{11},\qquad
u=\frac{q}{p+q}(1-e^{-w_0(p+q)\tau})\in(0,1).
\end{align*}
This excludes $\Phi_{d*}\rho\ge\varepsilon\sigma_d$ uniformly in $\rho$
for $\varepsilon>0$. Any QMS embedding would preserve $\sigma_d$ by
commutation with $\Phi_d$ and uniqueness of its invariant state.
Trace-norm contractivity then implies relaxation, hence positivity
improvement by
\cite[Theorem~4.2(6)--(7)]{FagnolaGirotti2026}, contradicting this
rank-two output. See \cite{SchnellEckardtDenisov2020} for Floquet embedding.
\end{remark}

\subsection{Continuous reservoir coupling}

For simultaneous coupling with common squeezing parameter $\theta$, set
\begin{align}
\cL_{d,\theta}&=\sum_{j=0}^{d-2}\mathcal G_j(\theta),
\qquad 0\le\theta\le1.
\end{align}

\begin{proposition}[Ordinary and complete MLSI bounds for a squeezed ladder]
\label{prop:squeezed-ladder}
The asymptotic conditional expectation is $E_d$. With
$g_*=\min\{(\sqrt p-\sqrt q)^2,q/2\}$,
\begin{align}
\alpha^{\mathrm c}(\cL_{d,\theta})
&\ge\frac{c_0g_*}{1+2\log Z_d+(d-1)\beta},
\\
\alpha(\cL_{d,\theta})
&\le\frac{(d-1)(p-q)\beta}{2dD_d}.
\end{align}
Both constants are of order $d^{-1}$, uniformly in $\theta$.
\end{proposition}

\begin{proof}
For the lower bound, use \eqref{eq:reservoir-coherences} with total
exit rates $b_i=\sum_l b_i^{(l)}$:
\begin{align*}
b_0&=q,\qquad b_{d-1}=p,\qquad
b_i=p+q\quad(1\le i\le d-2).
\end{align*}
On the GNS-orthonormal basis
$E_{jk}/\sqrt{\pi_k},E_{kj}/\sqrt{\pi_j}$,
the Hermitian part of $-\cL_{d,\theta}$ is
\begin{align*}
\begin{pmatrix}
(b_j+b_k)/2&-\theta(p+q)\mathbf1_{k=j+1}/2\\
-\theta(p+q)\mathbf1_{k=j+1}/2&(b_j+b_k)/2
\end{pmatrix}.
\end{align*}
Combining these blocks with the centered diagonal spectrum
\eqref{eq:reservoir-classical-gap} gives
\begin{align}
\lambda_{\GNS}(\cL_{d,\theta})
&=\min\left\{
p+q-2\sqrt{pq}\cos(\pi/d),\,
\frac{p+q}{2},\,
\frac{p+2q-\theta(p+q)}2
\right\}\nonumber\\
&\ge g_*>0.
\label{eq:squeezed-ladder-gap}
\end{align}
Hence $e^{t\cL_{d,\theta}}\to E_d$, and Corollary~\ref{cor:main}
with \eqref{eq:reservoir-index-entropy} proves the lower bound.

For the upper bound, test on $\Id/d$:
\begin{align*}
(\cL_{d,\theta})_*(\Id/d)
&=\frac{p-q}{d}(E_{00}-E_{d-1,d-1}),\\
\EP_{\cL_{d,\theta}}(\Id/d)
&=\Tr[(\cL_{d,\theta})_*(\Id/d)\log\sigma_d]
 =\frac{d-1}{d}(p-q)\beta,\\
\alpha(\cL_{d,\theta})
&\le\frac{\EP_{\cL_{d,\theta}}(\Id/d)}{2D_d}.
\end{align*}
Here $\log(\Id/d)$ is scalar and
$\Tr[(\cL_{d,\theta})_*(\Id/d)]=0$.
The order follows from $\alpha^{\mathrm c}\le\alpha$ and $D_d\asymp d$.
\end{proof}

The block estimate remains valid for edge-dependent $\theta_j\in[0,1]$.
For $\cL=\sum_jw_j\mathcal G_j(\theta_j)$,
$0<w_-\le w_j\le w_+$,
nonnegative local GNS forms give
\begin{align*}
-\Re\Tr[\sigma_dX^*\cL(X)]
&=-\sum_jw_j\Re\Tr[\sigma_dX^*\mathcal G_j(\theta_j)(X)]\\
&\ge w_-g_*\|X-E_dX\|_{2,\sigma_d}^2,\\
\EP_\cL(\Id/d)
&=\frac{(p-q)\beta}{d}\sum_jw_j
 \le\frac{w_+(d-1)(p-q)\beta}{d}.
\end{align*}
Thus the bounds acquire factors $w_-$ and $w_+$, respectively.

\begin{remark}[Symmetry and ideal squeezing]

Each local generator is KMS symmetric by the balance relation and
\begin{align*}
\sigma_d^{1/2}W_j^*\sigma_d^{-1/2}&=W_j,\qquad
[W_j^*W_j,\sigma_d]=0.
\end{align*}
For $\theta>0$, the unequal GNS matrix entries $\theta q,\theta p$
on adjacent coherences exclude GNS symmetry. At $\theta=0$,
\cite[Corollary~4.13 and Theorem~5.7]{GaoJungeLaRacuenteLi2025}
already gives $\alpha^{\mathrm c}\asymp d^{-1}$.

The reservoir parameters of \cite[Eq.~(6)]{LutkenhausCiracZoller1998} are
\begin{align*}
\gamma&=p-q,\qquad N=\frac{q}{p-q},\qquad
M=\frac{\theta\sqrt{pq}}{p-q},\\
\varphi&=\pi,\qquad M^2=\theta^2N(N+1).
\end{align*}
Here $\gamma$ is the damping rate, $N$ the occupation number, and
$M,\varphi$ the squeezing correlation and phase.
Thus $\theta=1$ is ideal squeezing.

The lower bound remains uniform up to $\theta=1$.
To compare with the unsqueezed generator, in the traceless
matrix-unit basis the Kossakowski blocks give, for $c\ge0$,
\begin{align*}
\cL_{d,\theta}-c\cL_{d,0}\text{ GKLS}
&\Longrightarrow
\begin{pmatrix}
(1-c)p&\theta\sqrt{pq}\\
\theta\sqrt{pq}&(1-c)q
\end{pmatrix}\ge0
\Longrightarrow c\le1-\theta.
\end{align*}
Thus this GKLS comparison degenerates at ideal squeezing.
At $d=2$, $\lambda_{\GNS}(\cL_{2,\theta})=(1-\theta)(p+q)/2$.
For $d\ge3$, overlapping transitions supply the additional dissipation
in \eqref{eq:squeezed-ladder-gap}.
\end{remark}
\begingroup
\let\originalbibitem\bibitem
\renewcommand{\bibitem}[1]{%
  \originalbibitem{#1}%
  \def\currentbibkey{#1}\def\lastbibkey{Wirth2026}%
  \ifx\currentbibkey\lastbibkey\raggedright\fi}
\bibliographystyle{amsplain}
\bibliography{references}
\endgroup

\end{document}